\documentclass[letterpaper, 10 pt, conference]{ieeeconf}  

\IEEEoverridecommandlockouts                              

\usepackage[usenames,dvipsnames]{xcolor}
\usepackage{psfrag,amsbsy,graphics,float}
\usepackage[dvips]{graphicx}
\usepackage{verbatim} 
\usepackage{amsmath} 
\usepackage{amssymb}  
\usepackage{algpseudocode} 
\usepackage{algorithm}
\usepackage{setspace}
\usepackage{balance}
\usepackage{booktabs}

\usepackage[section]{placeins}

\usepackage{pgfplots}
\usepackage{pgfplotstable}
\usepgfplotslibrary{fillbetween}

\usepackage{booktabs}
\pgfplotsset{compat = newest}
\usepackage{hyperref}

\newtheorem{mydef}{Definition}
\newtheorem{mytheo}{Theorem}
\newtheorem{mylemma}{Lemma}

\newtheorem{myproposition}{Proposition}
\newcommand{\R}{\mathbb{R}}
\newcommand{\N}{\mathbb{N}}
\newcommand{\C}{\mathcal{C}}
\newcommand{\X}{\mathcal{X}}
\newcommand{\U}{\mathcal{U}}
\newcommand{\bm}[1]{{\boldsymbol{#1}}}
\DeclareMathOperator{\diag}{diag}
\DeclareMathOperator{\var}{var}
\DeclareMathOperator{\Var}{\Sigma}
\DeclareMathOperator{\mean}{\mu}
\DeclareMathOperator{\expval}{E}

\newcommand{\GP}{\mathcal{GP}}

\usepackage{mathtools}
\mathtoolsset{showonlyrefs}

\title{\LARGE \bf
Equilibrium distributions and stability analysis of 
Gaussian Process State Space Models
}

\author{Thomas Beckers and Sandra Hirche
\thanks{T. Beckers and S. Hirche are with the Chair of Information-oriented Control (ITR), Department of Electrical and Computer Engineering,
Technical University of Munich, D-80333 Munich\newline
{\tt\small \{t.beckers, hirche\}@tum.de}}
}

\begin{document}

\maketitle
\thispagestyle{empty}
\pagestyle{empty}

\begin{abstract}
Gaussian Process State Space Models (GP-SSM) are a data-driven stochastic model class suitable to represent nonlinear dynamics. They have become increasingly popular in non-parametric modeling approaches since they provide not only a prediction of the system behavior but also an accuracy of the prediction. For the application of these models, the analysis of fundamental system properties is required. In this paper, we analyze equilibrium distributions and stability properties of the GP-SSM. The computation of equilibrium distributions is based on the numerical solution of a Fredholm integral equation of the second kind and is suitable for any covariance function. Besides, we show that the GP-SSM with squared exponential covariance function is always mean square bounded and there exists a set which is positive recurrent.
\end{abstract}

\section{Introduction}
The identification of a dynamical system plays a very import role in the area of control theory. The goal is the derivation of a mathematical model which is based on generated input data and the corresponding output data of a system. The model is necessary for any model-based control design, such as e.g. model predictive control. Besides, a model is required for simulations to evaluate control designs and to improve the understanding of the system. Classical system identification deals with parametric models, i.e. for linear dynamics ARX or ARMAX. If the system contains nonlinearities, there exist different identification techniques which mostly depends on the structure of the system. For these approaches, a suitable model structure must be selected to achieve useful results. Nevertheless, the identification of complex nonlinear systems with parametric models still poses a significant challenge. Especially, for complex systems such as human motion dynamics~\cite{wang2008gaussian} or the prediction of ozone concentration in the air~\cite{petelin2013evolving} non-parametric techniques appear to be more promising.\\
Within the past two decades, Gaussian Process regression has been used for modeling dynamical systems due to some beneficial properties such as the bias variance trade-off and the strong connection to Bayesian mathematics, see~\cite{frigola2014variational}. A Gaussian Process connects every point of a continuous input space with a normally distributed random variable. Any finite group of those infinitely many random variables follows a multivariate Gaussian distribution. Based on this, the result is a powerful tool for nonlinear function regression without the need of much prior knowledge~\cite{rasmussen2006gaussian}. In contrast to most of the other techniques, Gaussian Process modeling provides not only a mean function but also a measure for the uncertainty of the prediction. The output is a Gaussian distributed variable which is fully described by the mean and the variance.\\
The Gaussian Process State Space Model (GP-SSMs) uses this technique for modeling dynamical systems with state space models, see e.g.~\cite{kocijan2005dynamic}, where each state is described by its own GP. The model must be trained with input-output pairs of the system. Afterwards, the GP-SSM can predict the next step ahead state. Although the application of Gaussian Process State Space Models increases in control theory, e.g. for adaptive control~\cite{rogers2011adaptive}, the theoretical properties of the GP-SSM are only sparsely researched.\\
In most of the works where a GP-SSM is considered in a control setting only the mean function of the process is employed, see e.g.~\cite{wang2005gaussian} and~\cite{chowdhary2013bayesian}. This is mainly because the GP is often used for replacing an other deterministic method. In~\cite{beckers:ecc2016} some basic theoretical properties for deterministic GP-SSMs are derived. However, GP-SSMs contain a much richer description of the underlying dynamics but also the uncertainty about the model itself when the full stochastic representation is considered. In the stochastic case, also prediction uncertainty can be used to determine a suitable control law. For example, in~\cite{medina2015synthesizing} control laws are derived which explicitly take the uncertainty into account. In order to ensure the applicability of these control settings, classical control theory properties are required, see e.g.~\cite{kocijan2016modelling} and~\cite{avzman2008non}. Such basic properties of a dynamical system are among others the existence of equilibria and stability conditions.\\
For the control of GP-SSMs, knowledge about stochastic stability is essential. However, the calculation of equilibrium distributions of Gaussian Process State Space Models and the derivation of stability conditions are still open problems.\\
The contribution of this paper is the study of equilibrium distributions and stability of Gaussian Process State Space Models. We present a method to compute the equilibrium distribution which is based on the solution of a Fredholm integral equation. The method is usable for arbitrary covariance functions. For the widespread squared exponential covariance function, we present an upper bound in mean square sense and a set which is positive recurrent. We show that it is only possible to learn bounded systems with a GP-SSM with squared exponential covariance function. The derived results are illustrated in numerical simulations. 
\subsection{Notation}
Vectors and vector-valued functions are denoted with bold characters. Matrices are described with capital letters. The expression~$A_{:,i}$ denotes the i-th column of the matrix A. The expression~$\mathcal{N}(\mu,\Sigma)$ describes a normal distribution with mean~$\mu$ and covariance~$\Sigma$. The field of non-negative real numbers is denoted by~$\R_{\geq 0}$, positive real numbers by~$\R_{>0}$, and natural numbers without zero by~$\N_{>0}$. The Euclidean norm is given by~$\Vert\cdot\Vert$.
\section{Definitions}
This section starts with the necessary definitions of Gaussian Processes and Gaussian Process State Space Models.  
\subsection{GP Definition}
Let~$(\Omega, \mathcal{F},P)$ be a probability space with the sample space~$\Omega$, the corresponding~$\sigma$-algebra~$\mathcal{F}$ and the probability measure~$P$. The set~$\X \subseteq \R^n$ with~$n\in\N_{>0}$ denotes the index set. A stochastic process is a discrete or real valued function~$f(\bm x, \omega)$ which is a measurable function of~$\omega\in\Omega$ with~$\bm x\in\X$. The function~$f(\bm x, \omega)$ becomes a deterministic function if~$\omega\in\Omega$ is fixed which is called realization or sample part. In contrast, the function~$f(\bm x, \omega)$ is a random variable on~$\Omega$ when~$\bm x\in\X$ is specified and is denoted by~$f(\bm x)$. A Gaussian Process is such a stochastic process which is fully described by a mean function~$m(\bm x)\in\C^0$ and a covariance function~$k(\bm x,\bm x^\prime)\in\C^0$ since with fixed~$\bm x$ it is Gaussian distributed.
\begin{align}
f(&\bm x) \sim \GP(m(\bm x),k(\bm x,\bm x^\prime)),\qquad \bm x,\bm x^\prime\in\X\\
m(&\bm x)\colon\X\to\R,\qquad k(\bm x,\bm x^\prime)\colon\X\times \X\to\R
\end{align}
The mean function is usually defined to be zero, see~\cite{rasmussen2006gaussian}. The covariance function is a measure for the correlation of two states~$(\bm x,\bm x^\prime)$. The covariance function depends on so called hyperparameters whose number depends on the used function. The choice of the covariance function and the determination of the corresponding hyperparameters can be seen as degrees of freedom of the regression. Table \ref{tab:kernel} presents some common applied covariance functions.
\renewcommand{\arraystretch}{1.3}
\begin{table}[ht]
\hspace{1cm}
\begin{center}
\begin{tabular}{|m{1.7cm}|m{2.8cm}|m{2.9cm}|}
\toprule 
Name & Covariance function~$k(\bm x,\bm x^\prime)$ & Hyperparameters $\varphi=\{\ldots\}$\\ 
\midrule
linear &~$\bm x^\top \bm x^\prime+ \sigma_0^2$ &~$\sigma_0\in\R_+$ \\ 
polynomial &~$\left(\bm x^\top \bm x^\prime+ \sigma_0^2 \right) ^p$ &~$\sigma_0\in\R_+$\\ 
squared \mbox{exponential} &~$\sigma_f^2 \exp{\left(-\frac{\Vert \bm x- \bm x^\prime \Vert^2}{2l^2} \right) }$ &~$\sigma_f\in\R_{\geq 0},l\in\R_{>0}$\\ 
\bottomrule
\end{tabular} 
\end{center}
\caption{Summary of some commonly-used covariance functions.\label{tab:kernel}}
\end{table}
\renewcommand{\arraystretch}{1}\\
Probably the most widely used covariance function in Gaussian Process modeling is the squared exponential covariance function, see~\cite{rasmussen2006gaussian}, with the hyperparameters~$[l,\sigma_f]$. The length-scale~$l$ determines the number of expected upcrossing of the level zero in a unit interval by a zero-mean GP. The signal variance~$\sigma_f^2$ describes the average distance of the function~$f(\bm x)$ away from its mean. This covariance function is smooth, which leads to good results for modeling physical dynamics. An overview of the advantages of the different covariance functions can be found in~\cite{bishop2006pattern}.
\subsection{Gaussian Process State Space Models}
In this paper, we use discrete time Gaussian Process State Space Models with the~$n$-dimensional state vector~$\bm x_k\in\X$.
\begin{align}
\bm x_{k+1}&=\bm f(\bm x_k),\qquad k\in\N \label{for:GPSSS}\\
\bm y_k&=\bm x_k +\bm\varepsilon_k\\
\bm f(\bm x_k)&\sim \GP(\bm m(\bm x_k),\bm k(\bm x_k,\bm x^\prime_k))\\
\bm\varepsilon_k&\sim\mathcal{N}(\bm 0,\diag (\sigma_{1,n}^2,\ldots,\sigma_{n,n}^2))
\end{align}
Since the output of a Gaussian Process is one dimensional, a~$n$-dimensional system requires~$n$ GPs. Therefore, the vector valued function~$\bm m(\cdot)=[m_1(\cdot),\ldots,m_n(\cdot)]^\top$ describes the mean functions for each component of~$\bm x_{k+1}$. The Gaussian Process for each state depends on the corresponding mean and covariance function and is given by 
\begin{align}
\bm f(\bm x_k)&=
\begin{cases} 
f_1(\bm x_k)\sim \GP(m_1(\bm x_k),k_{\varphi_1}(\bm x_k,\bm x^\prime_k))\\
\vdots\hspace{0.9cm}\vdots\hspace{0.5cm}\vdots\\
f_n(\bm x_k)\sim \GP(m_n(\bm x_k),k_{\varphi_n}(\bm x_k,\bm x^\prime_k)).
\end{cases}
\end{align}
with the set of hyperparameters~$\varphi_i$. The GP-SSM has to be trained with an input and a corresponding output set. For this purpose, we arrange the~$m$ training inputs~$\{\tilde{\bm x}_{j}\}_{j=1}^m$ and corresponding outputs~$\{\tilde{\bm y}_{j+1}\}_{j=1}^m$ in an input training matrix~$X=[\tilde{\bm x}_{1},\tilde{\bm x}_{2},\ldots,\tilde{\bm x}_{m}]$ and an output training matrix~$Y^\top=[\tilde{\bm y}_{2},\tilde{\bm y}_{3},\ldots,\tilde{\bm y}_{m+1}]$. Therefore, the training data for the Gaussian Processes is described by~$\mathcal D=\{X,Y\}$. The prediction for each component $i$ of the one step ahead state vector~$x_{i,k+1}$ is calculated as Gaussian distributed variable with the conditional mean~$\mean(x_{i,k+1}\vert x_{i,k},\mathcal D)$ and the conditional variance~$\var(x_{i,k+1}\vert x_{i,k},\mathcal D)$. The joint distribution of the $i$-th component of the predicted next step ahead state~$x_{i,k+1}$ and the corresponding vector of the training outputs~$Y$ is 
\begin{align}
\begin{bmatrix} Y_{:,i} \\ x_{i,k+1} \end{bmatrix}\sim \mathcal{N} \left(\bm 0, \begin{bmatrix} K_{\varphi_i}(X,X) & \bm k_{\varphi_i}(\bm x_k,X)\\ \bm k_{\varphi_i}(\bm x_k,X)^\top & k_{\varphi_i}(\bm x_k,\bm x_k) \end{bmatrix}\right)\label{for:joint}
\end{align} 
where~$Y_{:,i}$ is the~$i$-th column of the matrix~$Y$. The function~$K_{\varphi_i}(X,X)$ is called covariance matrix and~$\bm k_{\varphi_i}(\bm x_k,X)$ the vector-valued extended covariance function with the set of hyperparameters~$\varphi_i$. 
\begin{align}
\begin{split}
&K_{\varphi_i}(X,X)\colon\X^m\times \X^m\to\R^{m\times m}\\
&K_{j',j}= k_{\varphi_i}(X_{:, j'},X_{:, j})\\
&\bm k_{\varphi_i}(\bm x_k,X)\colon\X\times \X^m\to\R^m,\,k_{j} = k_{\varphi_i}(\bm{x}_k,X_{:, j})\\
&\forall j',j\in\lbrace 1,\ldots,m\rbrace,i\in\lbrace 1,\ldots,n\rbrace .
\end{split}
\end{align}
With the assumption that the mean functions of the GPs are set to zero, a prediction of the~$i$-th component of~$\bm x_{k+1}$ is derived from the joint distribution~\eqref{for:joint}, see~\cite{rasmussen2006gaussian} for more details. This conditional probability distribution is Gaussian with the conditional mean
\begin{align}
\mean_i(\bm x_{k+1}\vert \bm x_k,\mathcal D)&=\bm k_{\varphi_i}(\bm x_k,X)^\top\bm h(i) \label{for:meanvalue}\\
\text{with }\bm h(i)& =(K_{\varphi_i}+I \sigma^2_{n,i})^{-1}Y_{:,i}\notag
\end{align}
where $\bm h(i)$ denotes the part which is independent of $\bm x_k$. The variance of the prediction is given by
\begin{align}
\var_i(\bm x_{k+1}\vert \bm x_k,\mathcal D)&=k_{\varphi_i}(\bm x_k,\bm x_k)-\bm k_{\varphi_i}(\bm x_k,X)^\top \notag\\
& \phantom{{}=}(K_{\varphi_i}+I \sigma^2_{i,n})^{-1} \bm k_{\varphi_i}(\bm x_k,X).\label{for:varvalue}
\end{align}
The variable~$\sigma_{i,n}\in\R$ is the standard deviation of the noise of the input data for all~$i\in\{1,\ldots,m\}$. The hyperparameters~$\varphi_i$ are optimized by means of the likelihood function, thus by maximizing the probability of 
\begin{align}
\varphi_i^* = \arg\max_{\varphi_i} \log P(Y_{:,i}|X,\varphi_i).
\end{align}
The~$n$ normally distributed components of $x_{i,k+1}\vert \bm x_k,\mathcal D$ are combined in a multi-variable Gaussian distribution 
\begin{align}
\bm x_{k+1}\vert \bm x_k,\mathcal D &\sim \mathcal{N} (\bm\mean(\cdot),\Var(\cdot))\\
\bm \mean(\bm x_{k+1}\vert \bm x_k,\mathcal D)&=[\mean_1(\cdot),\ldots,\mean_n(\cdot)]^\top\\
\Var(\bm x_{k+1}\vert \bm x_k,\mathcal D)&=\diag(\var_1(\cdot),\ldots,\var_n(\cdot)).
\end{align} 
Hence, the system~\eqref{for:GPSSS} can be rewritten as affine stochastic system with state depended noise
\begin{align}
\bm x_{k+1}=\bm \mean(\bm x_{k+1}\vert \bm x_k,\mathcal D)+\Var(\bm x_{k+1}\vert \bm x_k,\mathcal D)\bm\eta
\label{for:stochasticsystem}
\end{align} 
with the normally distributed random variable~$\bm \eta \sim \mathcal{N} \left(\bm 0,I\right)$.
\section{Equilibrium distribution}
The analysis of equilibrium points of stochastic systems requires first of all a definition of the stochastic equilibrium. If the variance is neglected, a deterministic approach can be used. To consider the stochastic behavior of the state variable, an equilibrium can be defined by an invariant distribution of the current state~$\bm x_k$ and the next state~$\bm x_{k+1}$.\\
Assume that the current state is a random variable~$\bm x_k$ with probability distribution~$p(\bm x_k)$. The predictive distribution is calculated by marginalizing over the state vector~\cite{girard2003gaussian}.
\begin{align}
p(\bm x_{k+1})=\int p(\bm x_{k+1}\vert \bm x_k,\mathcal D_i)p(\bm x_k)d \bm x_k
\label{for:preddist}
\end{align}
The probability distribution~$p(\bm x_{k+1}\vert \bm x_k,\mathcal D_i)$ is Gaussian 
\begin{align}
p(\bm x_{k+1}\vert \bm x_k,\mathcal D)=\mathcal{N}(\bm\mu(\bm x_{k+1}\vert\bm x_k,\mathcal D),\Var(\bm x_{k+1}\vert\bm x_k,\mathcal D))
\end{align}
with~\eqref{for:meanvalue} and~\eqref{for:varvalue}. An analytic solution of the integral is generally not possible but still obtained for some special cases, e.g. if the distribution~$p(\bm x_k)$ is also normal. Therefore, a solution for arbitrary distributions of~$\bm x_k$ can in general be found by numerical computation only. \\
To qualify as an equilibrium, the distribution of~$\bm x_k$ and~$\bm x_{k+1}$ must be equal. This condition transforms the predictive distribution equation into a linear, homogeneous Fredholm integral equation of the second kind.\\
The definition of this integral equation is given by 
\begin{align}
\underbrace{u(\bm x_{k+1})}_{p(\bm x_{k+1})}=\lambda\int{\underbrace{H(\bm x_{k+1},\bm x_k)}_{p(\bm x_{k+1}\vert \bm x_k,\mathcal D)}\underbrace{u(\bm x_k)}_{p(\bm x_k)}d\bm x_k}
\end{align}
where~$\lambda\in\R,\, \bm x_k,\bm x_{k+1}\in\X$ and~$H\colon\X\times\X\to\R$ are known piece-wise continuous functions while~$u\colon\X\to\R$ is an unknown function. The function~$H(\cdot)$ is known as kernel and~$\lambda$ is the eigenvalue.\\
In the following, we assume a one-dimensional system with~$\X=\R$. The extension to the multidimensional case is presented at the end of this section. A numerical solution of the integral equation can be found by using the Nystr\"om method, which approximates the integral by a finite sum, e.g. trapezoid rule, see~\cite{jerri1999introduction}. For this approach, the integral has to be defined on a finite interval~$[a,b]$ with~$a,b\in\R$ and the function~$H(x_{k+1},x_k)$ must be continuous. The length of the interval should be chosen sufficient large. The interval is divided in~$q$ equal parts of width~$\Delta x=\frac{b-a}{q}$. Additionally, let~$x_{i}=a+i\Delta x$ and~$x_{q}=b$.\\
The solution $u(\cdot)$ of the integral equation can be approximated by the matrix equation~$M\bm u=\bm 0$ with $M\in\R^{q\times q}$ and vector $\bm u\in\R^q$
\begin{align}
M&=\begin{bmatrix}
\frac{1}{\lambda}-\frac{\Delta x}{2}H_{0,0} & -\Delta x H_{0,1} & \ldots & -\frac{\Delta x}{2} H_{0,q}\\
-\frac{\Delta x}{2}H_{1,0} & \frac{1}{\lambda}-\Delta x H_{1,1} & \ldots & -\frac{\Delta x}{2} H_{1,q}\\
\vdots&\vdots&\ddots&\vdots\\
-\frac{\Delta x}{2}H_{q,0} & -\Delta x H_{q,1} & \ldots & \frac{1}{\lambda}-\frac{\Delta x}{2} H_{q,q}\\
\end{bmatrix}\hspace{0.5cm}
\label{for:intapprox}
\end{align}
where~$H_{i,j}=H(x_i,x_j)$ for~$i,j=0,1,\ldots,q$. The vector~$\bm u$ contains the approximation of the function values of~$u(\cdot)$ at~$x_i$. There exists an infinite number of non-zero solutions if and only if~$\det M=0$. This condition must be satisfied for~$\lambda=1$ to fulfill~\eqref{for:preddist}. Additionally,~$u(\cdot)$ must satisfy the constraints for a probability distribution.
\begin{align}
\int u(x_{k})dx_{k}&=1 \text{ and } u(x_{k})\geq 0,\,\forall x_{k}\in\R\label{for:constraint}
\end{align}
To find an appropriate solution, the linear equation~$M\bm u=\bm 0$ with the constraints~\eqref{for:constraint} must be solved. We use again the trapezoid rule to discretize the constraints
\begin{align}
&\Delta x\sum_0^{q} u_i-\frac{\Delta x}{2}(u_0+u_n)=1\label{for:const1}\\
&u_i\geq 0,\,\forall i=0,1,\dots,q
\end{align}
and add the constraint given by equation~\eqref{for:const1} to the matrix~$M$. The result is a non-homogeneous system of linear equations which can be formulated as least square optimization problem.
\begin{align}
\min_{\bm u}&\Vert M_p\bm u-\bm b_p\Vert^2\text{ with } u_i\geq 0,\,\forall i=0,1,\dots,q\\
M_p&=\begin{bmatrix}
&&M&&\\
\frac{\Delta x}{2} & \Delta x & \ldots & \Delta x & \frac{\Delta x}{2}
\end{bmatrix}\\
\bm b_p&=[0,\ldots,0,1]^T
\end{align}
If the residual of the optimization is sufficiently small, the vector~$\bm u$ is a discrete approximation of~$p(x_k)$ at
\begin{align}
x_k=a+i\Delta x\text{ for }i=0,1,\ldots,q
\end{align}
which solves equation~\eqref{for:preddist}.\\
If the system has more than one dimension, the numerical integration scheme for the Fredholm integral equation must be adapted. Generally, a numerical approximation for an integral of a continuous function~$g:D\to\R$ over a closed and bounded set~$D$ in~$\R^n$ is given by
\begin{align}
\int_D g(s)ds &\approx \sum_{i=0}^{q} w_{i} g(t_{i})
&\text{with }w_i\in\R,\,t_i\in D.
\end{align}
The used numerical approximation must converge to the true integral for~$q\to\infty$ to be valid for the presented algorithm, e.g. the multidimensional trapezoid rule satisfies that. With this approach, equation~\eqref{for:intapprox} is straightforward adopted to be applied in higher dimensional systems.\\
Algorithm (\ref{alg:u}) describes the whole computation in higher dimensional systems of the discrete approximation of~$p(\bm x_k)$.
\begin{algorithm}
\caption{Equilibrium Distribution \label{alg:u}}
\begin{algorithmic}
\State~$\bm b_p\leftarrow [0,\ldots,0,1]^T$
\For{q\,{$\{\bm x_{k+1}\}_i\in D$}}
	\For{q\,{$\{\bm x_{k}\}_j\in D$}}
	\State~$\bm \mean(\bm x_{k+1}) \leftarrow \bm k_{\varphi}^\top (K_{\varphi}+I\sigma^2_{n})^{-1}Y$
	\State~$\bm \var(\bm x_{k+1})  \leftarrow k_{\varphi}-\bm k_{\varphi}^\top (K_{\varphi}+I \sigma^2_{n})^{-1} \bm k_{\varphi}$
	\State~$H_{i,j}\leftarrow p(\bm x_{k+1}\vert \bm x_k,\mathcal D_i)$ 
	\EndFor
\EndFor
\State~$M=I-\text{weighted} \begin{pmatrix} H_{0,0} & \ldots & H_{0,q}\\ \vdots & \ddots & \vdots\\ H_{q,0} & \ldots & H_{q,q} \end{pmatrix}$
\If{$\det{M}\neq 0$}
\State \textbf{return No solution}
\Else
\State~$M_p=\left[\begin{array}{c}{M}\\ \text{normalized weights}\end{array}\right]$
\State~$\min_{\bm u}\Vert M_p\bm u-\bm b_p\Vert_2^2\text{ with } u_i\geq 0,\,\forall i=0,1,\dots,q$
\State \textbf{return}~$\bm u$
\EndIf
\end{algorithmic}
\end{algorithm}
\subsection{Remarks on convergence}
For a numerical approach it is important to analyze the convergence of the algorithm. The following proposition ensures that this condition is fulfilled.
\begin{myproposition}
Assume a finite interval~$[a,b]$ with boundaries~$a,b\in\R$ and a continuous solution~$p(x_{k+1})=p(x_k)$ of the integral equation
\begin{align*}
p(x_{k+1})&=\int_a^b p(x_{k+1}\vert x_k,\mathcal D_i)p(x_k)d x_k.
\end{align*}
The numerical solution~$p_q(x_k)$ given by the Nystr\"om method with the trapezoid rule converges to the exact solution~$p(x_k)$ if the step size~$\Delta x=\frac{b-a}{q}\to 0$ with~$q\to\infty$ and
\begin{align*}
\Delta x&\sum_{i=0}^{q} p_q(a+i\Delta x)-\frac{\Delta x}{2}(p_q(a)+p_q(b))\\
&=\int_a^b p(x_k)d x_k\\
&p_q(a+i\Delta x)\geq 0,\,\forall i=0,1,\dots,q.
\end{align*}
\label{prop:convergence}
\end{myproposition}
\begin{proof}
We start with the definition of the integral operator~$\mathcal{H}$ and the numerical integral operator~$\mathcal{H}_n$.
\begin{align}
\mathcal{H}p(x_{k+1})&=\int_a^b p(x_{k+1}\vert x_k) p(x_k)d x_k\\
\mathcal{H}_n p(x_{k+1})&=\sum_{i=0}^q w_i p(x_{k+1}\vert a+i\Delta x) p(a+i\Delta x)\\
&\text{with } x_{k+1}\in [a,b],w_i\in\R\notag
\end{align}
If the numerical integrator operator bases on the trapezoid rule with step size~$\Delta x=\frac{b-a}{q}$
\begin{align*}
\int_a^b g(x)dx\approx \Delta x \sum_{i=0}^q g(a+i\Delta x)-\frac{\Delta x}{2}(g(a)+g(b)),
\end{align*}
the numerical integral converges to the true integral for any continuous function~$g$, see~\cite{stoer2013introduction}. The speed of convergence of~$p_q(x_k)$ to the exact solution depends on the numerical integration error of the trapezoid rule.
\begin{align}
(\mathcal{H}-\mathcal{H}_n) p(x_k)&=-\frac{\Delta x^2}{12}\left[ \frac{\delta p(x_{k+1}\vert x_k)p(x_k)}{\delta x_k}\right]_{x_k=a}
^{x_k=b}\notag\\
&+O(\Delta x^4)
\end{align}
Since the difference~$(\mathcal{H}-\mathcal{H}_n) p(x_k)$ converges to zero for~$q\to\infty$ and~$\Vert p(x_k)-p_q(x_k)\Vert_\infty \leq c_s \Vert (\mathcal{H}-\mathcal{H}_n) p(x_k)\Vert$ with a constant~$c_s<\infty$, the numerical solution~$p_n(x_k)$ tends to the exact solution~$p(x_k)$, see~\cite{atkinson1997numerical}. 
\end{proof}
\section{Stability}
The previous section deals with the numerical computation of equilibrium distributions. Another important property of dynamical systems is stability. Several different stability measures exist for stochastic systems. This section makes use of the widespread mean square measure and positive recurrent sets. Since GP-SSMs are often used in combination with the squared exponential covariance function, the following stability analysis is focused on such kind of systems. We start with some definitions, see~\cite{kushner1971introduction}:
\begin{mydef}
A discrete-time dynamical system is called mean square bounded, if the solution~$\bm x_k~$ for~$k\in\N$ is bounded with~$\sup_{k\in\N} \expval \left[ \left\Vert \bm x_k \right\Vert^2 \right] < \infty$.
\end{mydef}
\begin{mydef}
The nonempty and measurable set~$\Lambda\subset\R^n$ is called positive recurrent if
\begin{align}
\sup_{\bm x_k\in\Lambda} \expval(\tau_\Lambda)<\infty
\end{align}
where~$\tau_\Lambda=\inf\left\lbrace k\geq 1\colon \bm x_k\in\Lambda\right\rbrace$ is the first return time to~$\Lambda$ if~$\bm x_0\in\Lambda$ and the first hitting time, otherwise.
\end{mydef}
\begin{mytheo} 
\label{theo:1}
A GP-SSM~\eqref{for:stochasticsystem} with squared exponential covariance function 
\begin{align*}
k_{\varphi_i}(\bm x,\bm x^\prime)=\sigma_{i,f}^2 \exp{\left(-\frac{\Vert \bm x- \bm x^\prime \Vert^2}{2l_i^2} \right) },\,\bm x,\bm x'\in\X
\end{align*}
where~$\sigma_{i,f}\in\R_{\geq 0}$ and $l_i\in\R_{>0}$ for all~$i\in\{1,\ldots,n\}$ with a number of $m$ training points is mean square bounded by
\begin{align*}
\sup_{k\in\N_{>0}} \expval \left[ \left\Vert \bm x_k \right\Vert^2 \right]\leq \sum_{i=1}^n \sigma_{i,f}^4 m \Vert \bm h(i)\Vert^2+\sigma_{i,f}^2.
\end{align*}
\end{mytheo}
Before starting the proof, we introduce important properties of the squared exponential covariance function~$k_{\varphi}(\bm x,\bm x^\prime)$.
\begin{mylemma}
For all~$\sigma_f\in\R_{\geq 0},l\in\R_{>0}$ the squared exponential covariance function is bounded, see \cite{beckers:ecc2016}, with
\begin{align}
\sup_{\bm x,\bm x^\prime\in\R^n}k_{\varphi}(\bm x,\bm x^\prime)&=\left. \sigma_{f}^2 \exp{\left(-\frac{\Vert \bm x-\bm x^\prime \Vert^2}{2l^2} \right)}\right|_{\bm x=\bm x^\prime}=\sigma_{f}^2 \label{for:covfuncbound}\\
\inf_{\bm x,\bm x^\prime\in\R^n}k_{\varphi}(\bm x,\bm x^\prime)&=\hspace{-0.3cm}\lim_{\Vert \bm x-\bm x^\prime\Vert\to\infty}\hspace{-0.2cm}\sigma_{f}^2 \exp{\left(-\frac{\Vert \bm x-\bm x^\prime \Vert^2}{2l^2} \right)}=0.
\end{align}
Therefore, the mean of the GP tends to zero and the variance is bounded when the states are far away from the training set. We use these properties for the following proof.
\label{lemma:bounds}
\end{mylemma}
\begin{proof}
We prove the mean square boundedness by evaluating the expected value~$\expval \left[\Vert \bm f(\bm x_{k})\Vert^2 \right]$ for each~$\bm x_{k}$. The expected value of a squared Gaussian distributed variable can be expressed by the addition of the squared mean and the variance.
\begin{align}
		\expval  \left[ \bm f^\top (\bm x_k) \bm f(\bm x_k)\right]=\sum_{i=1}^n &\mean_i^2(\bm x_{k+1}\vert \bm x_k,\mathcal D)\\
		+&\var_i(\bm x_{k+1}\vert \bm x_k,\mathcal D)\label{for:expfxfx}
\end{align}
If the squared mean~$\mean_i^2(\cdot)$ and the variance~$\var_i(\cdot)$ are bounded, then~\eqref{for:expfxfx} is bounded. The mean~$\mean_i(\bm x_{k+1}\vert \bm x_k,\mathcal D)$ is bounded with
\begin{align}
\Vert\mean_i(\bm x_{k+1}\vert \bm x_k,\mathcal D)\Vert&\leq \sigma_{i,f}^2\sqrt{m}\Vert \bm h(i)\Vert\\
\Rightarrow \mean_i^2(\bm x_{k+1}\vert \bm x_k,\mathcal D)&\leq \sigma_{i,f}^4 m \Vert \bm h(i)\Vert^2\\
\text{with }\bm h_i&=(K_{\varphi_i}(X,X)+I \sigma^2_{n,i})^{-1}Y_{:,i}.
\end{align}
by the application of the Cauchy-Schwarz inequality and lemma~\ref{lemma:bounds}. The variance
\begin{align}
\var_i(\bm x_{k+1}\vert \bm x_k,\mathcal D)&=k_{\varphi_i}(\bm x_k,\bm x_k)-\bm k_{\varphi_i}(\bm x_k,X)^\top \notag\\
& \phantom{{}=}(K_{\varphi_i}(X,X)+I \sigma^2_{n,i})^{-1} \bm k_{\varphi_i}(\bm x_k,X)
\end{align}
is also bounded by~$0\leq \var_i(\bm x_{k+1}\vert \bm x_k,\mathcal D) \leq \sigma_{i,f}^2$ because of lemma~\ref{lemma:bounds} and the positive definiteness of the matrix~$(K_{\varphi_i}(X,X)+I \sigma^2_{n,i})^{-1}$. Therefore, the solution~$\bm x_k$ with~$k>0$ of system~\eqref{for:stochasticsystem} is mean square bounded with
\begin{align}
\sup_{k\in\N_{>0}} \expval \left[ \left\Vert \bm x_k \right\Vert^2 \right]\leq \sum_{i=1}^n \sigma_{i,f}^4 m \Vert \bm h(i)\Vert^2+\sigma_{i,f}^2. \label{for:Eupperbound}
\end{align}
\label{proof:meansquarebounded} 
\end{proof}
This theorem can be interpreted as follows: Since the mean and the variance of $\bm x_{k+1}\vert \bm x_k$ are bounded it is only possible to learn bounded systems with a GP-SSM with squared exponential covariance function. This upper bound depends on the signal variance $\sigma_f$ and noise variance $\sigma_n$, the number of training points $m$, and their position. The value of the upper bound increases if the number of training points  or the values of the output training data $Y$ increase.\\
Now, we want to focus on the behavior of the trajectories of the system. For this purpose, we use the theory of Markov chains because the future state of the system \eqref{for:stochasticsystem} only depends on the current state and thus it is Markovian.
\begin{mytheo} 
For a GP-SSM~\eqref{for:stochasticsystem} with squared exponential covariance function 
\begin{align*}
k_{\varphi_i}(\bm x,\bm x^\prime)=\sigma_{i,f}^2 \exp{\left(-\frac{\Vert \bm x- \bm x^\prime \Vert^2}{2l_i^2} \right) },\,\bm x,\bm x'\in\X
\end{align*}
with~$\sigma_{i,f}\in\R_{\geq 0},l_i\in\R_{>0}$ for all~$i\in\{1,\ldots,n\}$, there exists a set
\begin{align*}
\Lambda=\{\bm x \in\X\vert \left\Vert \bm x \right\Vert^2\leq \sum_{i=1}^n \sigma_{i,f}^4 m \Vert \bm h(i)\Vert^2+\sigma_{i,f}^2\},
\end{align*}
which is positive recurrent.
\end{mytheo}
\begin{proof}
First, we recall the criterion for positive recurrent sets. Positive recurrency guarantees that the system trajectory returns to a set in a finite time horizon.
\begin{mylemma}[\cite{kushner1971introduction}]
Suppose that there exists a positive definite Lyapunov function $V(\bm x)$ and positive constants~$c_1,c_2,c_3\in\R_{>0}$ such that
\begin{align}
\expval\left[V(\bm x_{k+1}\vert \bm x_k)\right]-V(\bm x_k)&\leq -c_2,&\text{if } V(\bm x)>c_1\\
\expval\left[V(\bm x_{k+1}\vert \bm x_k)\right]-V(\bm x_k)&\leq c_3<\infty,&\text{if } V(\bm x)\leq c_1
\end{align}
Then the set
\begin{align}
\Lambda=\left\lbrace \bm x \colon V(\bm x)\leq c_1 \right\rbrace
\end{align}
is positive recurrent.
\end{mylemma}
We assume the positive definite Lyapunov function
\begin{align}
V(\bm x)=\bm x^\top \bm x,\qquad\bm x\in\R^n.
\end{align} 
The drift of $V(\bm x)$ is given by
\begin{align}
		\Delta V&=\expval \left[ V(\bm x_{k+1}\vert\bm x_k)\right]-V(\bm x_k) \notag\\
		&=\expval \left[ \bm f^\top (\bm x_k) \bm f(\bm x_k)\right] - \bm x_k^\top \bm x_k.
\end{align}
An upper bound for~$\left[ \bm f^\top (\bm x_k) \bm f(\bm x_k)\right]$ is given by equation~\eqref{for:Eupperbound} which results in
\begin{align}
		\Delta V&\leq\sum_{i=1}^n \sigma_{i,f}^4 m \Vert \bm h(i)\Vert^2+\sigma_{i,f}^2-\bm x_k^\top \bm x_k.\label{for:deltav}
\end{align}
Due to the fact that~$\lim_{\left\Vert \bm x_k \right\Vert\to\infty} \bm x_k^\top \bm x_k=\infty$ is unbounded and equation~\eqref{for:deltav}, there must exist a set~$\Lambda$ with a neighbourhood~$\U=\R^n\backslash \left\lbrace\Lambda\right\rbrace$ which fulfills
\begin{align}
	 \Delta V<0,\,\bm x_k\in \U.
\end{align}
The drift~$\Delta V$ is negative if~$\bm x_k^\top \bm x_k>\expval \left[ \bm f^\top (\bm x_k) \bm f(\bm x_k)\right]$. Therefore, the set~$\Lambda$ is defined by
\begin{align}
\Lambda=\{\bm x \in\R^n\vert \left\Vert \bm x \right\Vert^2\leq \sum_{i=1}^n \sigma_{i,f}^4 m \Vert \bm h(i)\Vert^2+\sigma_{i,f}^2 \}.
\end{align}
Since the drift of the Lyapunov function is negative outside the set $\Lambda$, lemma 2 is fulfilled and thus the set is positive recurrent. 
\end{proof}
\section{Simulations}
\subsection{Equilibrium Distribution}
In this section, we present an examples of equilibrium distributions of a one-dimensional Gaussian Process State Space Model with squared exponential covariance function. The solution is validated by a Monte Carlo experiment and a two-sample Kolmogorov-Smirnov test.\\
We assume a system which is described by
\begin{align}
x_{k+1}=0.01x_k^3-0.2x_k^2+0.2x_k+\eta\label{for:nonlinearsys}
\end{align}
where~$\eta$ is standard normal distributed. A Gaussian Process State Space Model with squared exponential covariance function is trained with 20 input points which are uniformly distributed on the interval $[-5,5]$ and the corresponding output data. The output data is corrupted by a Gaussian noise with a variance of $\sigma_n=1$. The hyperparameters are optimized by maximizing the marginal likelihood. The optimized value of the lengthscale $l$ is $3.59$ and the signal noise $\sigma_f$ is $4.21$.\\
The predicted mean function (blue) and the variance (gray) of the GP-SSM are drawn in Fig.~\ref{fig:figure3}.
   \begin{figure}[htp]
\begin{tikzpicture}
\begin{axis}[
  xlabel={$x_k$},
  ylabel={$x_{k+1}$},
  line width=1pt,
  axis equal=true,
  grid style={dashed,gray},
  grid = both,
  xmin=-9.1, xmax=3, ymin=-11.1, ymax=5]
\addplot+[name path=varp, color=gray, no marks] table [x index=0,y index=2]{data/figure3_mean_var.dat};
\addplot+[name path=varm, color=gray, no marks] table [x index=0,y index=3]{data/figure3_mean_var.dat};
\addplot[gray,opacity=0.5] fill between[ of = varm and varp];
\addplot+[color=blue, no marks,line width=1.2pt] table [x index=0,y index=1]{data/figure3_mean_var.dat};
\addplot+[mark=+,color=green!70!black, only marks,mark size=4] table [x index=0,y index=1]{data/figure3_points.dat};
\addplot[const plot,fill=orange,draw=black] table [x index=0,y expr=\thisrowno{1}*60-11]{data/figure3_hist.dat};
\addplot[const plot, fill=blue!50!red!50!white, draw=black] table [x expr=\thisrowno{3}*60-12.5,y index=2]{data/figure3_hist.dat};
\addplot[color=black, no marks,line width=1.2pt] table [x index=0,y expr=\thisrowno{1}*60-11]{data/figure3_px.dat};
\addplot[color=black, no marks,line width=1.2pt] table [x expr=\thisrowno{1}*60-12.5,y index=0]{data/figure3_px.dat};
\end{axis}
\end{tikzpicture} 
      \normalfont{\vspace{-0.2cm}\caption{A GP-SSM with squared exponential function trained by 20 noisy data points (green crosses) of the nonlinear system~\eqref{for:nonlinearsys}. The GP Regression gives the resulting mean function (blue) and variance (gray). The black line at the bottom and on the left side of the figures describes the computed equilibrium distribution. A Monte Carlo experiment with the input samples (orange bars) based on the equilibrium distribution and the output samples (purple bars) supports that the distribution is an equilibrium.}\label{fig:figure3}}
   \end{figure}
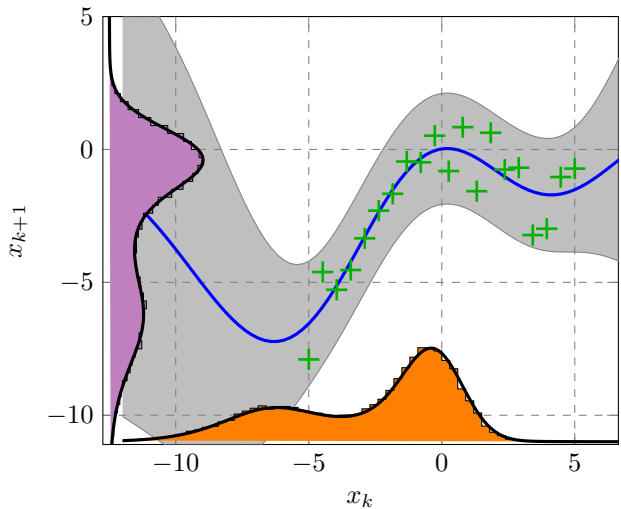
The equilibrium distribution (black line) is non-Gaussian shaped which is based on the strong nonlinear behavior of the mean function. To cover the relevant array of the distribution, we set the interval of the integral to~$[-12,8]$ and divide it in~$q=150$ parts. The determinant of~$M$ is zero and the active-set algorithm find a minimum for~$\Vert M_p\bm p-\bm b_p\Vert^2$ with holding the constrains. The optimization results in a residual value of~$2.3\text{\sc{e}-}6$. To validate the computed distribution, we use the inverse transform sampling method to generate 30000 sample out of this distribution. Since the inverse of the cumulative distribution function is necessary for the inverse transform sampling method, the discrete points~$p_i$ of the probability density function are numerically integrated. These samples are visualized by the orange bars at the bottom of Fig.~\ref{fig:figure3}. The purple bars on the left side show the output distribution of the samples. \\
The two-sample Kolmogorov-Smirnov test returns that it is not possible to reject the null hypothesis that the probability distribution are identical  at the 5\% significance level. The Monte Carlo experiment and the Kolmogorov-Smirnov test support the assumption that the calculated distribution function~$p(x_k)$ is a equilibrium distribution of the nonlinear discrete-time system.
\vspace{0.2cm}
\subsection{Stability}
This example shows the boundedness of the GP-SSM with squared exponential covariance function. We use the highly nonlinear Van der Pol oscillator as training system. The discretization of the oscillator is described by~\cite{van2010new} with 
\begin{align}
	x_{k+1}&=\phi(T,x_k,y_k,\epsilon) \Psi(x_k,y_k)T\notag\\
	&+(\varphi(T,x_k,y_k,\epsilon)+1)x_k+n_1\notag\\
	y_{k+1}&=\phi(T,x_k,y_k,\epsilon) \Lambda(x_k,y_k)T\notag\\
	&+(\varphi(T,x_k,y_k,\epsilon)+1)y_k+n_2
	\label{for:vanderpol}
\end{align}
where the sample time~$T$ is set to~$0.1$ and the parameter~$\epsilon$ to~$-0.8$. The functions $\phi(\cdot)$ and $\varphi(\cdot)$ are highly nonlinear which generates a non-conservative oscillator with nonlinear damping.\\
 A Gaussian distributed noise~$n_1,n_2\sim\mathcal{N}(0,0.01^2)$ is added to the output data set. The GP-SSM is trained with 441 uniformly distributed points on the square~$[-3,3]\times[-3,3]$. The hyperparameters are optimized by the minimization of the log-likelihood function with a conjugate gradient method. For the multi-step ahead prediction not only the mean but also the uncertainty is considered, see~\cite{girard2003gaussian}. Since the trajectory stays inside the training area, the predicted trajectory is very similar. The mean square boundedness of the trained GP-SSM is fulfilled.\\
The model is tested with two different set of initial points~$x_0$ and~$y_0$. For~$x_0=-1.8,y_0=0$, Fig.~\ref{fig:fig4} shows the trajectory of the system~\eqref{for:vanderpol} and the mean~$\bar{x}_k,\bar{y}_k$ with the~$2\sigma$ standard deviation of the multi-step ahead prediction of the trained GP-SSM. The predicted mean and the trajectory of~\eqref{for:vanderpol} are quite similar.\\
For the second example, the initial state of the system is changed to~$x_0=2.2,y_0=0$ which generates an unstable trajectory, see Fig.~\ref{fig:fig5}. Due to the fact that this initial point is not inside the attraction area of the Van der Pol oscillator, the trajectory~$x_k,y_k$ of the system is not bounded. Nevertheless, the GP-SSM generates a bounded mean and variance function. This test case is just done to demonstrate the boundedness of the GP-SSM. The increased variance shows the uncertainty of the prediction since the model can not generate the unstable trajectory.
   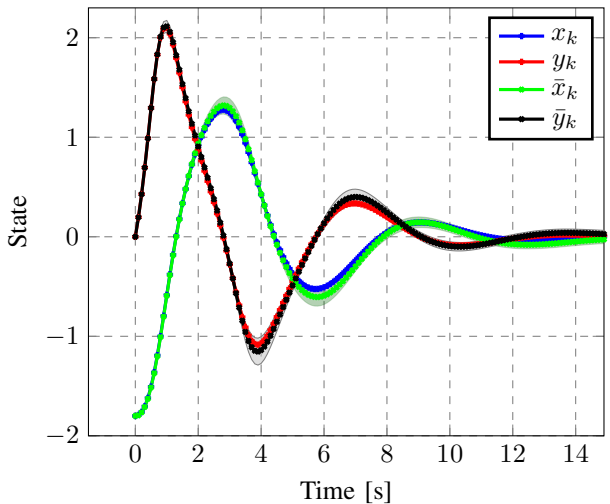
\begin{figure}[ht]
   \vspace{0.22cm}
	\begin{tikzpicture}
\begin{axis}[
  xlabel={Time [s]},
  ylabel=State,
  line width=1pt,
  grid style={dashed,gray},
  grid = both,
  ymin=-2,
  ymax=2.3,
  xmax=14.9]
  \addplot+[name path=varp1, color=black!60!green,opacity=0.2, no marks] table [x index=0,y expr=\thisrowno{1}+\thisrowno{3}]{data/figure4_mean_std.dat};
\addplot+[name path=varm1, color=black!60!green,opacity=0.2, no marks] table [x index=0,y expr=\thisrowno{1}-\thisrowno{3}]{data/figure4_mean_std.dat};
\addplot[black!60!green,opacity=0.3] fill between[ of = varm1 and varp1];
  \addplot+[name path=varp2, color=gray, no marks] table [x index=0,y expr=\thisrowno{2}+\thisrowno{4}]{data/figure4_mean_std.dat};
\addplot+[name path=varm2, color=gray, no marks] table [x index=0,y expr=\thisrowno{2}-\thisrowno{4}]{data/figure4_mean_std.dat};
\addplot[lightgray,opacity=0.5] fill between[ of = varm2 and varp2];
\addplot[mark=+,mark size=1.2,color=blue,line width=1.2pt] table [x index=0,y index=1]{data/figure4_points.dat};
\addplot[mark=+,mark size=1.2,color=red,line width=1.2pt] table [x index=0,y index=2]{data/figure4_points.dat};
\addplot[mark=x,color=green,mark size=1.2,line width=1.2pt] table [x index=0,y index=1]{data/figure4_mean_std.dat};
\addplot[mark=x,color=black, mark size=1.2,line width=1.2pt] table [x index=0,y index=2]{data/figure4_mean_std.dat};
\legend{,,,,,,$x_k$,$y_k$,$\bar{x}_k$,$\bar{y}_k$};
\end{axis}
\end{tikzpicture} 
      \normalfont{\vspace{-0.0cm}\caption{The mean~$\bar{x}_k,\bar{y}_k$ and the~$2\sigma$ standard deviation of the multi-step ahead prediction by a GP-SSM with squared exp. covariance function is always bounded. With~$x_0=-1.8,y_0=0$ the predicted mean and the trajectory of~\eqref{for:vanderpol} are quite similar. The variance of the prediction is low.\label{fig:fig4}}} 
      \vspace{-0.1cm}
   \end{figure}
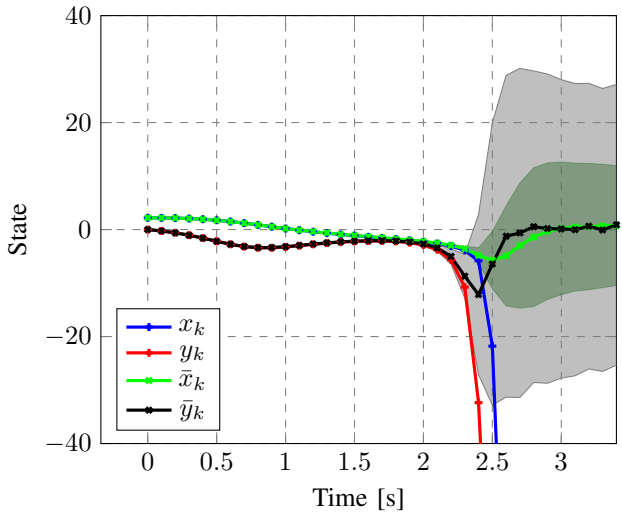
\begin{figure}[ht]
\begin{tikzpicture}
\begin{axis}[
  xlabel={Time [s]},
  ylabel=State,
  legend pos=south west,
  grid style={dashed,gray},
  grid = both,
    ymin=-40,
  ymax=40,
  xmax=3.4]
\addplot+[name path=varp1, color=black!60!green,opacity=0.3, no marks] table [x index=0,y expr=\thisrowno{1}+\thisrowno{3}]{data/figure5_mean_std.dat};
\addplot+[name path=varm1, color=black!60!green,opacity=0.3, no marks] table [x index=0,y expr=\thisrowno{1}-\thisrowno{3}]{data/figure5_mean_std.dat};
\addplot[black!60!green,opacity=0.5] fill between[ of = varm1 and varp1]; 
\addplot+[name path=varp3, color=gray, no marks] table [x index=0,y expr=\thisrowno{2}+\thisrowno{4}]{data/figure5_mean_std.dat};
\addplot+[name path=varm3, color=gray, no marks] table [x index=0,y expr=\thisrowno{2}-\thisrowno{4}]{data/figure5_mean_std.dat};
\addplot[gray,opacity=0.5] fill between[ of = varm3 and varp3];
\addplot[mark=+,mark size=1.5,color=blue,line width=1.2pt] table [x index=0,y index=1]{data/figure5_points.dat};
\addplot[mark=+,mark size=1.5,color=red,line width=1.2pt] table [x index=0,y index=2]{data/figure5_points.dat};
\addplot[mark=x,color=green,mark size=1.5,line width=1.2pt] table [x index=0,y index=1]{data/figure5_mean_std.dat};
\addplot[mark=x,mark size=1.5,color=black,line width=1.2pt] table [x index=0,y index=2]{data/figure5_mean_std.dat};
\legend{,,,,,,$x_k$,$y_k$,$\bar{x}_k$,$\bar{y}_k$};
\end{axis}
\end{tikzpicture} 
      \normalfont{\vspace{-0.0cm}\caption{The prediction of the mean~$\bar{x}_k,\bar{y}_k$ and the corresponding~$2\sigma$ standard deviation of a GP-SSM with squared exponential covariance function is always bounded even if the trajectory~$x_k,y_k$ of the original system is unbounded. For testing purpose, the GP-SSM should generate an unbounded trajecory. Since the GP-SSM is bounded, the trajectory of the true system is not reproduced.\label{fig:fig5}}}
      \vspace{-0.1cm}
\end{figure}
\section*{Conclusion}
\balance
In this paper, we present fundamental control properties of GP-SSMs from a stochastic point of view. In the first part, an algorithm for the computation of equilibrium distributions for Gaussian Process State Space Models is shown. The method bases on the solution of a Fredholm integral equation which is done by numerical approximation. The result is a system of linear equations and constraints to ensure that the solution is a valid probability distribution.\\
The second part deals with the proof of the mean square boundedness of a GP-SSM with squared exponential covariance function. We also show that there exists a set which is positive recurrent. Therefore, it is only possible to learn bounded systems with a GP-SSM with squared exponential covariance function. The computation of equilibrium distributions is validated in a simulation which uses input sample points that are generated of the equilibrium distribution. A simulation of a discrete Van der Pol oscillator shows the mean square boundedness.

\section*{ACKNOWLEDGMENTS}
The research leading to these results has received funding from the European Research Council under the European Union Seventh Framework Program (FP7/2007-2013) / ERC Starting Grant ``Control based on Human Models (con-humo)'' agreement n\textsuperscript{o}337654.

\bibliography{mybib}
\bibliographystyle{ieeetr}

\end{document}